\documentclass[10pt, conference, letterpaper]{IEEEtran}

\IEEEoverridecommandlockouts
\usepackage[utf8]{inputenc}
\usepackage[english]{babel}
\usepackage{cite}
\usepackage{amsmath,amssymb,amsfonts}
\usepackage{nicefrac}
\usepackage{textcomp}
\usepackage{xcolor}
\usepackage{multirow} %
\usepackage{tabularx} %
\usepackage{balance}
\usepackage{booktabs}
\usepackage{threeparttable} %
\usepackage[nolist]{acronym}
\def\BibTeX{{\rm B\kern-.05em{\sc i\kern-.025em b}\kern-.08em
    T\kern-.1667em\lower.7ex\hbox{E}\kern-.125emX}}
\usepackage{paralist}

\usepackage{graphicx}
\graphicspath{{./figures/}}

\usepackage{tikz}
\usepackage{mathtools}
\usepackage{amsthm}
\usepackage[hidelinks]{hyperref}
\usepackage{cleveref} %
\usepackage{csquotes} %
\usepackage{subcaption}
\usepackage{newfloat}
\makeatletter
\newenvironment{ilp}[1]
 {%
  \addtocounter{equation}{-1}%
  \crefalias{subequations}{ilp}
  \begin{subequations}%
  \def\@currentlabel{#1}%
  \renewcommand{\theequation}{#1.\alph{equation}}%
 }
 {\end{subequations}}
\makeatother %
\DeclareFloatingEnvironment[fileext=frm,placement={tbp},name=Problem]{problem}
\Crefname{problem}{Problem}{Problems}
\def\BibTeX{{\rm B\kern-.05em{\sc i\kern-.025em b}\kern-.08em
    T\kern-.1667em\lower.7ex\hbox{E}\kern-.125emX}}

\newtheorem{theorem}{Theorem}
\newtheorem{lemma}{Lemma}

\theoremstyle{remark}
\newtheorem*{proofsketch}{Proof Sketch}

\DeclareMathOperator*{\argmax}{arg\,max}
\DeclareMathOperator*{\argmin}{arg\,min}

\DeclarePairedDelimiter\ceil{\lceil}{\rceil}
\DeclarePairedDelimiter\floor{\lfloor}{\rfloor}
\let\union\cup

\renewcommand{\epsilon}{\varepsilon}

\newcommand{\ecap}{\ensuremath{\textit{cap}}}
\newcommand{\ccap}{\ensuremath{\textit{ccap}}}
\newcommand{\cost}{\ensuremath{\textit{cost}}}

\newcommand{\mcfalpha}{\varrho}
\newcommand{\ssep}{\mid}

\newcommand{\bigO}{\ensuremath{\mathcal{O}}}
\newcommand{\Q}{\ensuremath{\mathbb{Q}}}

\newcommand{\N}{\ensuremath{\mathbb{N}}}
\newcommand{\Npos}{\ensuremath{\mathbb{N}_{\geq1}}}
\newcommand{\be}{\ensuremath{\coloneqq}}
\newcommand{\transpose}[1]{\ensuremath{#1^{\top}}}

\newcommand{\flow}{\ensuremath{f}}
\newcommand{\flowsum}{\ensuremath{\mathbf{f}}}
\newcommand{\flowset}{\ensuremath{\mathcal{F}}}
\newcommand{\demand}{\ensuremath{T}}

\newcommand{\lpopt}{\ensuremath{x^*}}

\newcommand{\algval}{\ensuremath{z}}
\newcommand{\lpval}{\ensuremath{z^*}}
\newcommand{\ilpval}{\ensuremath{z^*_\mathit{int}}}

\newcommand{\paredges}[1]{\ensuremath{\lambda_{#1}}}
\newcommand{\meanpar}{\ensuremath{\paredges{\text{avg}}}}

\newcommand{\minpar}{\ensuremath{\paredges{\text{min}}}}
\newcommand{\bidir}[1]{\ensuremath{\{#1\}}}
\newcommand{\edges}{\ensuremath{E(A)}}
\newcommand{\idx}{\ensuremath{\mathit{idx}}}
\newcommand{\alldemand}{\ensuremath{\mathbb{T}_A}}

\newcommand{\prob}[1]{\ac{#1}}

\usepackage{algorithm}
\usepackage[noend]{algpseudocode}
\newcommand*\Let[2]{\State #1 $\gets$ #2}
\algnewcommand\Break{\State \textbf{break}}
\algnewcommand\Continue{\State \textbf{continue}}
\algdef{SE}[DOWHILE]{Do}{DoWhile}{\algorithmicdo}[1]{\algorithmicwhile\ #1}%
\algrenewcommand\Return{\State \algorithmicreturn{} }
\algrenewcommand\algorithmicforall{\textbf{for each}}
\algnewcommand\algorithmicinput{\textbf{Input:}}
\algnewcommand\algorithmicoutput{\textbf{Output:}}
\algnewcommand\Input{\item[\algorithmicinput]}%
\algnewcommand\Output{\item[\algorithmicoutput]}%
\algrenewcommand\algorithmicthen{}
\algrenewcommand\algorithmicdo{}
\usepackage{stackengine}
\setstackEOL{\\}

\usetikzlibrary{arrows,backgrounds,shapes,calc,decorations,decorations.pathreplacing,decorations.pathmorphing,positioning}
\tikzset{main_edge/.style={line width=2.3pt}}
\tikzset{dot/.style={circle,fill=black,inner sep=1pt,minimum size=1pt}}
\tikzset{path/.style={
    line join=round,
    decorate, decoration={
        zigzag,
        segment length=6,
        amplitude=.7,
        post=lineto,
        pre length=3pt,
        post length=3pt
    }, very thick,
    shorten <= 3pt,
    shorten >= 3pt
}}
\tikzset{edge/.style={
    rounded corners, -stealth, draw=black, very thick, shorten <= 2pt,shorten >= 2pt}}
\tikzset{biedge/.style={
    rounded corners, draw=black, shorten <= 2pt,shorten >= 2pt}}

\newcommand\copyrighttext{%
  \footnotesize \textcopyright 2026 IEEE. Personal use of this material is permitted.
  Permission from IEEE must be obtained for all other uses, in any current or future
  media, including reprinting/republishing this material for advertising or promotional
  purposes, creating new collective works, for resale or redistribution to servers or
  lists, or reuse of any copyrighted component of this work in other works.}
\newcommand\copyrightnotice{%
\begin{tikzpicture}[remember picture,overlay]
\node[anchor=south,yshift=10pt] at (current page.south)
  {\fbox{\parbox{\dimexpr\textwidth-\fboxsep-\fboxrule\relax}{\copyrighttext}}};
\end{tikzpicture}%
}

\newsavebox{\jamBox}
\newlength{\jamWidth}
\newcommand{\jamIfToBig}[2]{%
    \savebox{\jamBox}{#2}%
    \settowidth{\jamWidth}{\usebox{\jamBox}}%
    \ifthenelse{\jamWidth < #1}%
        {\usebox{\jamBox}}%
        {\resizebox{#1}{!}{\usebox{\jamBox}}%
    }%
}

\makeatletter
\AtBeginDocument
{
\def\ltx@label#1{\cref@label{#1}}%
\def\label@in@display@noarg#1{\cref@old@label@in@display{#1}}%
\def\label@in@mmeasure@noarg#1{%
\begingroup%
    \measuring@false%
    \cref@old@label@in@display{#1}%
\endgroup}%
}%
\makeatother

\usepackage[textsize=tiny]{todonotes}

\begin{document}

\begin{acronym}
\acro{DAG}{directed acyclic graph}
\acro{DEFO}{Declarative and Expressive Forwarding Optimizer}
\acro{ECMP}{Equal Cost Multipath}
\acro{IGP}{Interior Gateway Protocol}
\acro{IE}{Ingress-Egress}
\acro{ISP}{Internet Service Provider}
\acro{IS-IS}{Intermediate System to Intermediate System}
\acro{LDP}{Label Distribution Protocol}
\acro{LER}{Label Edge Router}
\acro{LP}{linear program}
\acro{LSR}{Label Switched Router}
\acro{LSP}{Label Switched Path}
\acro{MCF}{Multi-Commodity Flow}
\acro{MLU}{Maximum Link Utilization}
\acro{MO}{Midpoint Optimization}
\acro{MPLS}{Multiprotocol Label Switching}
\acro{MSD}{Maximum Segment Depth}
\acro{NDA}{non-disclosure agreement}
\acro{OSPF}{Open Shortest Path First}
\acro{PoP}{Point of Presence}
\acro{RSVP}{Resource Reservation Protocol}
\acro{RL2TLE}{Router-Level 2TLE}
\acro{SDN}{Software-Defined Networking}
\acro{SC2SR}{Shortcut 2SR}
\acro{SID}{Segment Identifier}
\acro{SPR}{Shortest Path Routing}
\acro{SR}{Segment Routing}
\acro{SRLS}{Segment Routing Local Search}
\acro{TE}{Traffic Engineering}
\acro{TLE}{Tunnel Limit Extension}
\acro{WAE}{WAN Automation Engine}
\acro{w2TLE}{Weighted 2TLE}
\acro{ILP}{integer linear program}
\acro{RSP}{Route Switch Processor}
\acro{NPU}{Network Processing Unit}
\acro{TOCA}{\textsc{Traffic-Oblivious Connection Activation}}
\acro{TAS}{\textsc{Traffic-Aware Subnetwork}}
\acro{MMCFS}{\textsc{Minimum Multi-Commodity Flow Subgraph}}
\acro{TSP}{\textsc{Travelling Salesperson Problem}}
\acro{MWS3}{\textsc{Min-Weight-Sat(3)}}
\acro{DHC}{\textsc{Directed Hamilton Cycle}}
\end{acronym}

\title{No Traffic to Cry: Traffic-Oblivious\\ Link Deactivation for Green Traffic Engineering
\thanks{Funded by the German Research Foundation (DFG)
grant 461207633 (projects AS 341/7-2, CH 897/7-2).}%
}

\author{\IEEEauthorblockN{%
Max Ilsen\IEEEauthorrefmark{1},
Daniel Otten\IEEEauthorrefmark{1},
Nils Aschenbruck\IEEEauthorrefmark{1},
Markus Chimani\IEEEauthorrefmark{1}}
\IEEEauthorblockA{%
\IEEEauthorrefmark{1}Osnabrück University, Institute of Computer Science, Osnabrück, Germany\\
Email: \{max.ilsen, daotten, aschenbruck, markus.chimani\}@uos.de}}
\maketitle

\copyrightnotice
\begin{abstract}
As internet traffic grows, the underlying infrastructure consumes increasing amounts of energy. During off-peak hours, large parts of the networks remain underutilized, presenting significant potential for energy savings. Existing Green Traffic Engineering approaches attempt to leverage this potential by switching off those parts of the networks that are not required for the routing of specific traffic matrices. When traffic changes, the approaches need to adapt rapidly, which is hard to achieve given the complexity of the problem.
We take a fundamentally different approach: instead of considering a specific traffic matrix, we rely on a traffic-oblivious routing scheme. We discuss the NP-hard problem of activating as few connections as possible while still guaranteeing that \emph{any} down-scaled traffic matrix~\( \mcfalpha\cdot\demand\) can be routed, where~\(\mcfalpha \in (0,1)\) and~\(\demand\) is any traffic matrix routable in the original network. We present a \(\max(\frac{1}{\mcfalpha\cdot\minpar},2)\)-approximation algorithm for this problem, with \(\minpar\) denoting the minimum number of connections between any two connected routers. Additionally, we propose two post-processing heuristics to further improve solution quality.
Our evaluation shows that we can quickly generate near-optimal solutions.
By design, our method avoids the need for frequent reconfigurations and offers a promising direction to achieve practical energy savings in backbone networks.
\end{abstract}


\section{Introduction}
Constantly growing internet traffic leads to growing \ac{ISP} infrastructure and energy consumption. This does not only increase operational costs but also leads to massive $CO_2$ emissions. For instance, Telefonica reported that in 2023 they consumed 41 MWh per petabyte of traffic and emitted 337\,119 tons of  $CO_2$ \cite{telefonica}.
A major part of this energy is used to maintain a network infrastructure that
can handle the maximum daily amount of traffic.
However, as shown by \Cref{FIG:2022-traffic}, for nearly a third of the day, the amount of traffic is below 50\% of
this maximum~\cite{DBLP:conf/lcn/SchullerACHS17,The_internet}, with the average link
utilization even dropping below 20\%~\cite{DBLP:conf/infocom/HassidimRSS13}.

\begin{figure}
		\centering
		\includegraphics[width=.83\linewidth]{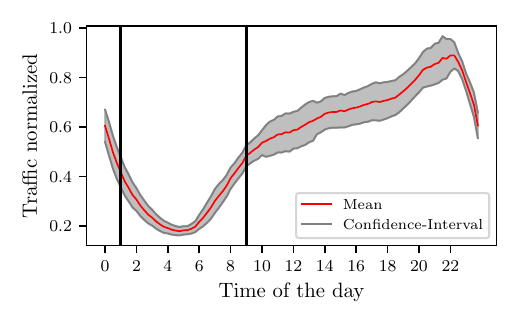}
		\caption{Traffic on an average day in 2022}
		\label{FIG:2022-traffic}
\end{figure}

In spite of the huge potential for energy savings, no Green \ac{TE} approach has yet made a practical impact. Existing work in this field generally follows a similar pattern: given a specific network topology and one or more traffic matrices, the algorithms remove as many components as possible while still ensuring that the remaining network can handle the traffic.
This pattern faces two key limitations: First, the algorithms need to react quickly to changes in traffic or other unforeseen events. However, many \ac{TE} problems are known to be NP-hard, making fast reconfiguration difficult or impossible in practice~\cite{INFOCOM:Approximation_Algs,DBLP:conf/cp/HartertSVB15}.
Second, even if a new configuration is computed in time, the router hardware and software must be capable of implementing these changes quickly and reliably. This is far from trivial, especially in large-scale backbone networks with strict reliability requirements.

Together, these two constraints limit the applicability of existing solutions. To overcome them, a fundamentally different approach is needed: one that can guarantee that the remaining topology is able to carry a wide set of traffic variations.
In this paper, we present such an approach and support it with a strong theoretical foundation. Our main contributions are:
\begin{compactitem} %
    \item We define the NP-hard problem of activating as few connections
    as possible while still allowing for a multi-commodity flow routing of
    \emph{any} down-scaled traffic matrix~$\mcfalpha\cdot\demand$ with~$\mcfalpha\in (0,1)$, where~$\demand$
    is only restricted in so far that is must be routable in the network~$G$ when all connections are active.
    \item  We develop a~$\max(\frac{1}{\mcfalpha\cdot\minpar},
    2)$-approximation algorithm with $\minpar$ as the minimum number of connections
    between any two connected routers.
    \item We show that in practical scenarios, our algorithms generate
    solutions that are only slightly worse than the optimum,
    usually requiring a significantly lower running time and yielding a lower \acl{MLU}.
\end{compactitem}

\section{Related Work}\label{sec:relatedWork}
There is a rich selection of recent papers dealing with \ac{TE} in general. Most of them focus on improvements to already existing routing methods---for example, in \cite{MO_INFOCOM,INFOCOM:MulticasTE}, the respective authors build upon the more theoretical \acl{SR} concept presented in~\cite{DBLP:conf/infocom/BhatiaHKL15} and try to bridge the gap between research and practice.
Other authors derive general but mainly theoretical results on related routing problems \cite{INFOCOM:Optimal_Oblivious,INFOCOM:Designing_Optimal_Compact,INFOCOM:Efficient_Algorithm,INFOCOM:Approximation_Algs}. However, only few papers have been published in the field of Green \ac{TE} in recent years \cite{DBLP:conf/lcn/OttenICA23,OttenLCN24,CATE,Hypnos}; the vast majority of papers in this field is much older. Green \ac{TE} approaches mainly differ in two aspects:

\textbf{1) The switched-off hardware:} Almost all Green \ac{TE} approaches aim to reduce energy consumption by offloading and switching off certain hardware components. However, there are considerable differences depending on the power model used. Some authors consider the deactivation of entire routers~\cite{DBLP:conf/ifip6-3/AddisCCGS12,DBLP:conf/wowmom/AmaldiCGM11,DBLP:conf/icc/ChiaraviglioMN09,DBLP:journals/cn/JiaJGSW18}. Other authors try to minimize the number of active connections in the network, which results in fewer active ports overall~\cite{Identifying_energy_critical_paths, green_sr, GreyWolf,DBLP:conf/ict/GhumanN17, DBLP:journals/cn/JiaJGSW18, DBLP:journals/ijcomsys/OkonorWGS17, DBLP:conf/icnp/ZhangYLZ10,Hypnos}. Still others propose approaches that lie somewhere in between: they aim to deactivate bundles of connections such that specific hardware components of routers (e.g.\ linecards), rather than the entire router, can be turned off \cite{DBLP:conf/lcn/OttenICA23,OttenLCN24}. Nevertheless, what all approaches have in common is that they downscale the network; they all require the removal of connections. Hence, basic research must consider this as an objective function. Starting from this, solutions can easily be adapted to specific hardware requirements by summarizing connections or using a weight function.

\textbf{2) \ac{TE} technology:}
One can use different \ac{TE} technologies to steer traffic away from the hardware components that should be switched off. Especially older approaches often still rely on \ac{IGP} metric tuning \cite{DBLP:conf/icccnt/OkonorGA16,OSPF_enhancement,GOSPF,An_Energy_Saving_Routing_Algorithm}. While this is a long-standing and well-understood \ac{TE} concept, it has various limitations (e.g., its inability to precisely control the paths of individual demands). Thus, metric tuning has been replaced by newer and more sophisticated \ac{TE} technologies. One of these, which is frequently used in the Green \ac{TE} context as well, is \ac{MPLS} with \ac{RSVP}-\ac{TE}, see \cite{Identifying_energy_critical_paths, green_sr, GreyWolf, DBLP:conf/ict/GhumanN17, DBLP:journals/cn/JiaJGSW18, DBLP:journals/ijcomsys/OkonorWGS17, DBLP:conf/icnp/ZhangYLZ10, DBLP:journals/jocnet/SankaranS14,DBLP:conf/icc/ChiaraviglioMN09}. Some authors, for example \cite{DBLP:conf/icnp/ZhangYLZ10}, do not explicitly name a \ac{TE} technology, but their approaches would require such a tool if implemented. However, with networks growing more and more in size, the rather limited scalability of \ac{MPLS} is becoming an increasing issue. Thus, it is (again) being replaced by more flexible recent technologies like \acl{SR}, as in \cite{DBLP:conf/lcn/OttenICA23,OttenLCN24}.

In summary, while existing Green \ac{TE} approaches explore different technologies and power models, they all follow the same basic pattern of disabling hardware components
to save energy.
However, as discussed in the introduction,
such approaches face fundamental challenges %
in terms of computational complexity and practical deployability when faced with variations in traffic.
We address this issue by building upon theoretical results concerning network design for \emph{traffic-oblivious} \ac{MCF} routing.

One field of related work in this direction is that of \emph{flow sparsifiers},
which are graphs that preserve the congestion of an input graph up to a certain
factor~\cite{DBLP:conf/focs/Moitra09,DBLP:conf/stoc/LeightonM10,DBLP:conf/soda/AndoniGK14}---however,
flow sparsifiers may include entirely new edges and thus cannot be used to model
the simple deactivation of connections.
More applicable to our use case is work on \emph{robust network design}~\cite{ben2005routing,DBLP:journals/jcss/AzarCFKR04,DBLP:journals/talg/HajiaghayiKRL07,DBLP:conf/stacs/Al-NajjarBL22}, %
where one usually minimizes the cost of reserving
capacity on the edges while establishing the routability of specific traffic matrices that are given in the form of a polytope.
We focus on a recent publication~\cite{DBLP:conf/isaac/ChimaniI25} as a starting point: it presents an
approximation algorithm for finding a directed subgraph with a minimum number of arcs
that allows for an \ac{MCF} routing of any traffic routable in the original
network when scaled down by a certain factor.

\section{Model and Problem Definition}\label{sec:model}
\subsection{Graph-based Model}\label{sec:graph_model}

For our theoretical contributions, we will consider general
directed graphs (or \emph{digraphs})~$G=(V,A)$ with positive arc capacities~$\ecap$
to showcase the broad applicability of our results.
However, when discussing real networks, we will always focus on the special case of
bidirected graphs where for every arc~$uv \in A$, there exists an opposing
arc~$vu \in A$ with the same capacity~$\ecap(vu) = \ecap(uv)$.
Such bidirected graphs form the basis of our network model, with the
nodes~$V$ representing the routers and each arc~$uv \in A$ representing the
the total data transmission capability between two routers \emph{in one
direction}: from~$u$ to~$v$.

Note that in real-world networks a link between two routers is comprised of a bundle of bidirected connections,
and the two directions belonging to the same connection cannot be turned off
separately.
Thus, given a bidirected graph~$G=(V,A)$ with unidirectional \emph{arcs}, we will also define a set of bidirected \emph{edges}~$\edges \be
\{\bidir{u,v} \ssep uv \in A\}$ with the capacity~$\ecap(\bidir{u,v}) \be \ecap(uv)$ for each edge~$\bidir{u,v}\in \edges$.
Further, each edge~$e\in\edges$ represents a bundle of $\paredges{e} \in \Npos$~many
bidirected connections that can be turned on or off.
We let~$\ccap(e)\be\frac{\ecap(e)}{\paredges{e}}$ denote the capacity of a single bidirected connection in the bundle
represented by~$e$.

Our results can easily be extended to multi-graphs with multiple parallel
arcs that share the same two nodes as source and target respectively, but have different
capacities.
To simplify our notation, we do not consider multi-graphs for the remainder of this paper.

\subsection{Routing Strategies}\label{sec:routing}
The traffic in the network is given by a \emph{traffic matrix}~$\demand$:
for every ordered pair of nodes~$(s,t)$ in the network, $\demand(s, t)$ specifies the
corresponding \emph{demand}, i.e., how many units of traffic have to be sent from~$s$
to~$t$.
We may refer to such a pair of nodes as a \emph{commodity}.
Given a digraph~$G=(V,A)$ with capacities~$\ecap$ and a traffic
matrix~$\demand$, a \emph{flow}~$\flow_{s,t} \colon A \to \Q$ from a node~$s$ to a
node~$t$ is a function that satisfies the flow conservation constraints:
for all nodes~$v \in V$,
\begin{align*}
    \sum_{uv \in A} \flow_{s,t}(uv) - \sum_{vu \in A} \flow_{s,t}(vu) &\; = \;
            \begin{cases*}
                -\demand(s,t) & if $v=s$,\\
                \demand(s,t) & if $v=t$,\\
                0 & else.
            \end{cases*}
\end{align*}
A \emph{routing strategy}~$X$ then is a set of flows~$\{\flow_{s,t} \ssep (s,t)\in
V^2\}$ that satisfies the capacity constraints: for all arcs~$uv \in A$,
\begin{align*}
    \sum_{(s,t) \in V^2} \flow_{s,t}(uv) & \leq \ecap(uv).
\end{align*}
We call a traffic matrix~$\demand$ $X$-routable in a capacitated
graph~$(G,\ecap)$ if there exists a routing strategy~$X$ for~$\demand$
in~$(G,\ecap)$, and consider three types of routing strategies:

\begin{asparaitem}
    \item\textbf{\acf{SPR}}:
        Each edge is assigned a weight (according to a
        \emph{routing metric}), and each demand is routed along its shortest
        path w.r.t.\ to these weights.
        In case of non-unique shortest paths we use \ac{ECMP}-routing,
        meaning that whenever the shortest paths for a commodity diverge, the
        corresponding flow is split equally among these diverging paths.
        In practice, \ac{SPR} is realized using the \ac{IGP} \ac{OSPF}.
    \item\textbf{$k$-\acf{SR}}:
        Each commodity~$(s,t)$ is assigned $k-1$~midpoints~$v_1,\dots,v_{k-1}$
        and routed via~\ac{SPR} from~$s$ to~$v_1$, from each midpoint to the
        subsequent midpoint and finally from~$v_{k-1}$ to~$t$.
        The number of these routing \emph{segments}~$k$ is usually constrained
        by hardware.
        We only consider \emph{node segments} (rather than adjacency segments)
        as this is common in the Green \ac{TE} literature and allows for a simpler
        optimization model.
        This means that an \ac{SR}~path only consists of nodes~$v_i \in V$
        and can be realized via a stack of router SSIDs that is added to the
        packet at the ingress node.
    \item\textbf{\acf{MCF}}: Apart from the flow conservation and capacity
        constraints, the routing is unconstrained.
        In general, routing via arbitrary paths is not practically viable due to
        hardware constraints.
        Thus, \ac{MCF} is typically mainly considered to establish lower bounds on routability and link utilization.
        Interestingly, in real-world scenarios $2$-\ac{SR} attains solutions very close to \ac{MCF}~\cite{DBLP:conf/infocom/BhatiaHKL15}.
\end{asparaitem}

For each routing strategy named above, there exists a more restricted
\emph{unsplittable} variant, which only allows choosing a single path for every
commodity.

\subsection{Complexity Considerations}\label{sec:complexity}

Choosing a routing strategy to model the routing in backbone networks requires
careful consideration.
Note that regardless of the routing strategy, finding a subnetwork of
minimum size that supports this routing is typically NP-hard---even in a
traffic-aware scenario where a specific traffic matrix is already given.
In many cases, finding such a subnetwork is even NP-hard to approximate:
Consider the problem where, given a digraph $G = (V,A)$ with both capacities and
costs on its arcs as well as a traffic matrix~$\demand$, the task is to find a
minimum-cost subdigraph~$G'$ of~$G$ such that there exists an unsplittable
\ac{SPR} satisfying all capacity constraints and demands given
by~$\demand$.
Bley~\cite[Corollary 4.7]{DBLP:conf/ipco/Bley05} shows that this is NP-hard to
approximate within a factor of~$2^{N^{(1-\varepsilon)}}$ where~$N$ is the
encoding size of the instance and $\varepsilon > 0$ is arbitrarily small.
While the proof focuses on unsplittable shortest paths, it entails that the same
holds w.r.t.\ any arbitrary unsplittable routing scheme (including unsplittable
$k$-\ac{SR} and unsplittable \ac{MCF}).
We can also extend the proof to allow any constant number~$\ell$ of flow splits
by copying the digraph constructed in the proof~$\ell$ times and adding
super-source and super-sink nodes connected to all original source and sink
nodes, respectively.
Lastly, we give a similar result for the problem variant
where one asks for a subgraph with a minimum number of arcs
(rather than one with minimum cost), which we will call \prob{TAS}:

\begin{theorem}\label{th:tas_inapproximability}
    It is NP-hard to approximate \prob{TAS} within any polynomial factor.
\end{theorem}
\begin{proofsketch} %
    Reducing from \acl{DHC}, it is well-known that the \prob{TSP} with polynomially bounded arc lengths is NP-hard
    to approximate within any polynomial factor.
    Using \prob{TSP} as the original problem in the reduction from \cite[Theorem~3.1]{1990OrponenApprox}
    then yields the same for \prob{MWS3} with polynomially bounded variable costs.
    We can then give a reduction from \prob{MWS3} to \prob{TAS}
    by building on top of the construction from~\cite[Corollary~4.7]{DBLP:conf/ipco/Bley05}:
    Subdivide those arcs that are shared by multiple clause gadgets, and assign each clause to a separate subdivision arc.
    Then, each arc with a positive
    cost~$c$ is replaced by a path of length $c\cdot R\log R$ where~$R$ is the number
    of those remaining arcs that are necessarily contained in any feasible
    solution.
    As $c$ is polynomially bounded, the reduction still runs in polynomial time. \hfill\qedsymbol
\end{proofsketch}

Moreover, we highlight the difficulty of evaluating the quality of an
already computed network w.r.t.\ \ac{SPR} and $k$-\ac{SR}:
Given a digraph and a traffic matrix, it is NP-hard to choose arc weights that
minimize the \ac{MLU} of the corresponding \ac{SPR} within a factor of
1.5~\cite{DBLP:journals/coap/FortzT04}.
Already for 2-\ac{SR}, the problem remains hard even when the arc weights are
\emph{fixed}: given a digraph with (possibly even unit) arc weights and a traffic matrix, it is NP-hard
to find a corresponding 2-\ac{SR} that respects the capacity constraints~\cite{DBLP:conf/cp/HartertSVB15}.
This makes it difficult to evaluate the quality of a given subgraph~$G'$
since just calculating the \ac{MLU} obtainable via 2-\ac{SR} in~$G'$ would
already require solving an NP-hard problem.
In comparison, calculating an \ac{MCF} for a given traffic matrix in a given
network is optimally solvable in polynomial time using linear programs and yields routing paths that,
on real-world instances, can equivalently be realized via
2-\ac{SR}~\cite{DBLP:conf/infocom/BhatiaHKL15}.
Thus, we will use \ac{MCF} as our algorithmic quality evaluation oracle when minimizing the subnetwork, and investigate the practical soundness of this choice in our experiments.

\subsection{Problem Definition}\label{subsec:Problem}

We define the \prob{TOCA} problem:
As we consider a traffic-oblivious scenario, we are only given a network
topology in the form of a bidirected graph~$G=(V,A)$ with arc
capacities~$\ecap \colon A \to \Q$ and $\paredges{e}$~connections for each bidirected edge~$e \in \edges$, as well as
a retention ratio~$\mcfalpha\in(0,1)$ that represents how much the
traffic scales down during low-traffic periods.
While it may not scale down evenly across the network, the traffic is at least upper bounded by such an even down-scaling.
A solution for \prob{TOCA} is a function
that maps each edge~$e \in \edges$ to a number of connections~$x_e \in
\{0,\dots,\paredges{e}\}$ that should be active (the remaining $\paredges{e} -
x_e$ connections are to be turned off during low-traffic periods).
The solution is feasible if for every traffic matrix~$\demand$ that is
\ac{MCF}-routable in $(G,\ecap)$, the scaled-down traffic matrix
$\mcfalpha\cdot\demand$ is \ac{MCF}-routable in~$(G,\ecap')$  where
$\ecap'(uv)\be x_{\bidir{u,v}}\cdot\ccap(\bidir{u,v})$ for every arc~$uv\in A$.
That is, every such scaled-down traffic matrix~$T$ must still be
\ac{MCF}-routable when only the capacity of $x_e$~connections for each edge~$e \in \edges$
is available.
Our task is to find a feasible solution with the minimum number of
connections~$\sum_{e \in \edges} x_e$.
In line with \cite{DBLP:conf/isaac/ChimaniI25}, we call the variant of \prob{TOCA} where the input
graph is directed (not necessarily bidirected) with one connection per arc
\prob{MMCFS}.
This variant is already NP-hard~\cite{DBLP:conf/isaac/ChimaniI25}.

Interestingly, \cite[Theorem 3]{DBLP:conf/isaac/ChimaniI25} shows that for \prob{MMCFS}, instead of
considering all possible \ac{MCF}-routable traffic matrices, it suffices to
consider a single \enquote{worst-case} traffic matrix which defines
a demand of value~$\ecap(st)$ exactly for
each directed arc~$st$ in the digraph~$G=(V,A)$:
\begin{equation}\label{eq:super_t_arcs}
    \alldemand(s,t) \be
    \begin{cases*}
        \ecap(st) & if $st \in A$, \\
        0 & otherwise.
    \end{cases*}
\end{equation}

We generalize this reformulation to the case where an arc might not
only be turned on and off entirely but also be assigned an arbitrary reduced
capacity.
This generalization clearly includes the \ac{TOCA} setting with multiple
connections per bidirected edge where an arbitrary subset of the connections of an edge may be
active.
The reformulation is still applicable when all flows are unsplittable.

\begin{theorem}\label{th:super_t_unsplittable}
    Given a directed (possibly bidirected) graph~$G=(V,A)$ with arc capacities~$\ecap$, a retention
    ratio~$\mcfalpha\in(0,1)$, and a reduced capacity function~$\ecap'$
    with~$\ecap'(a) \leq \ecap(a)$ for all~$a \subseteq A$,
    the following statements are equivalent:
    \begin{itemize}
        \item For all traffic matrices~$T$ that are \ac{MCF}-routable
            in~$(A,\ecap)$, the scaled matrix~$\mcfalpha\cdot\demand$ is
            \ac{MCF}-routable in~$(A,\ecap')$.
        \item The scaled matrix~$\mcfalpha\cdot\alldemand$ is \ac{MCF}-routable
            in~$(A,\ecap')$.
    \end{itemize}
    This also holds when we restrict all mentions of \ac{MCF}-routability in the
    above statements to unsplittable \ac{MCF}.
\end{theorem}
\begin{proof}
    $\alldemand$ is routable in $(A,\ecap)$ by definition.
    If every traffic matrix routable in $(A,\ecap)$ is also routable
    in~$(A,\ecap')$ when scaled down by~$\mcfalpha$, then so is
    $\mcfalpha\cdot\alldemand$.

    For the other direction, consider any arbitrary traffic matrix~$T$ routable
    in~$(A,\ecap)$.
    Let $\{\flow^{\demand}_{s,t} \ssep (s,t) \in V^2\}$ be the (possibly
    unsplittable) \ac{MCF} that routes~$\demand$ in~$(A,\ecap)$ with the vector
    $\flowsum^\demand \be \sum_{(s,t) \in V^2} \flow^\demand_{s,t}$ specifying
    the total flow over each edge.
    Using this \ac{MCF}, we can construct a new traffic matrix~$\demand'$ with
    $\demand'(s,t) \be \flowsum^\demand(st)$ if $st \in A$, and 0 otherwise.

    Using component-wise comparison, we have $\demand'\leq\alldemand$.
    Thus, since $\mcfalpha\cdot\alldemand$ is routable in $(A,\ecap')$, so is
    $\mcfalpha\cdot\demand'$.
    But if $\mcfalpha\cdot\demand'$ is routable in~$(A,\ecap')$ using the
    flows~$\{\flow^{\mcfalpha\cdot\demand'}_{u,v} \ssep (u,v) \in V^2\}$, then
    $\mcfalpha\cdot\demand$ is also routable in~$(A,\ecap')$ using the
    flows~$\{\flow^{\mcfalpha\cdot\demand}_{s,t} \ssep (s,t) \in V^2\}$
    constructed as follows:
    for each commodity~$(s,t) \in V^2$, and each arc~$uv\in A$, calculate
    the fraction of flow routed over $uv$ that is used by $\flow^\demand_{s,t}$,
    and route this fraction over the path chosen by
    $\flow^{\mcfalpha\cdot\demand'}_{u,v}$, i.e.,
    \begin{equation*}
        \flow^{\mcfalpha\cdot \demand}_{s,t}(e) \be
            \sum_{uv\in A \colon \flowsum^\demand(uv) > 0}
            \frac{\flow^\demand_{s,t}(uv)}{\flowsum^\demand(uv)}\cdot\flow^{\mcfalpha\cdot\demand'}_{u,v}(e). \qedhere
    \end{equation*}
\end{proof}

\section{Approximation via LP Rounding}\label{sec:approximation}

\subsection{ILP and LP Relaxation}

In \cite{DBLP:conf/isaac/ChimaniI25}, the authors present an ILP (Integer Linear Program) for \prob{MMCFS} and show that
rounding up a basic optimal solution for the LP relaxation of this ILP is a
$\max(\frac{1}{\mcfalpha}, 2)$-approximation for \prob{MMCFS}.
We adapt their approach to our practically motivated graph-based model outlined
in \Cref{sec:model}, while improving upon the aforementioned approximation
guarantee. Building upon their ILP, we obtain ILP~\eqref{eq:ilp_mcfs_undirected}.
\begin{figure*}
\begin{ilp}{1}\label{eq:ilp_mcfs_undirected}
\begin{align}
    \min &\sum_{e \in \edges}x_e & \label{eq:ilp_bidir_obj}\\
    \sum_{u\colon uv \in A} \flow_{s,t}(uv) - \sum_{u \colon vu \in A} \flow_{s,t}(vu) &=
            \begin{cases*}
                -\mcfalpha\alldemand(st) & if $v=s$\\
                \mcfalpha\alldemand(st) & if $v=t$\\
                0 & else
            \end{cases*}
        &\forall v\in V, st \in A \label{eq:ilp_bidir_flow_conservation}\\
    \sum_{st \in A} \flow_{s,t}(uv) &\leq x_{\bidir{u,v}} \cdot \ccap(\bidir{u,v})
                       &\forall uv\in A \label{eq:ilp_bidir_cap_constraint}\\
    \flow_{s,t}(uv) &= \flow_{t,s}(vu)
                        &\forall st\in A, uv\in A \label{eq:ilp_bidir_flow_symmetry}\\
    \flow_{s,t}(uv) &\geq 0 &\forall st \in A, uv \in A\\
    x_e & \in \{0,\dots,\paredges{e}\} &\forall e \in \edges \label{eq:ilp_bidir_integrality_constraint}
\end{align}
\end{ilp}
\caption*{ILP~\eqref{eq:ilp_mcfs_undirected}: ILP formulation for \prob{TOCA}}
\end{figure*}
Rather than one variable for each directed arc, we use one $x$-variable
for each bidirected edge in~$\edges$ (since the two directions of a connection cannot be
turned off separately).
Further, such a variable for an edge~$e$ is not binary but
integer valued, i.e.,~$x_e \in \{0, \dots, \paredges{e}\}$, which indicates how
many connections in the corresponding bundle are active.
However, we keep separate flow variables for each directed arc in order to
model the flow.
The optimization goal~\eqref{eq:ilp_bidir_obj} is to minimize the total number
of active connections.
The flow conservation constraints~\eqref{eq:ilp_bidir_flow_conservation} ensure
that the $\flow$-variables represent correct flows and that all
demands~$\alldemand$ are satisfied.
Moreover, the capacity constraints~\eqref{eq:ilp_bidir_cap_constraint} are
adapted such that the total flow over a directed arc~$uv$---which we may
denote by $\flowsum(uv) \be \sum_{st\in A} \flow_{s,t}(uv)$---does not
surpass the capacity of active connections from~$u$ to~$v$.
Lastly, we also add constraint~\eqref{eq:ilp_bidir_flow_symmetry}, which ensures that the flow
for commodity~$(s,t)\in V^2$ is routed over the same but reversed paths as that for~$(t,s)$.
This constraint may seem superfluous and possibly restricting at first glance. However, it plays
a crucial role to attain our approximation guarantee, to soundly capture the extension that we can now
only decide on edges instead of individual arcs.
This additional constraint also does not change the set of feasible solutions w.r.t.\ to the $x$-variables:

\begin{theorem}\label{th:bidir_flow_symmetry}
    Let~$G = (V,A)$ be a bidirected graph with arc
    capacities~$\ecap$ and~$\demand$ a traffic matrix.
    $\demand$ is routable in~$(G,\ecap)$ via an \ac{MCF}~$\flowset = \{\flow_{s,t} \ssep
    (s,t) \in V^2\}$ if and only if it is routable in~$(G,\ecap)$ via an
    \ac{MCF}~$\flowset' = \{\flow'_{s,t} \ssep (s,t) \in V^2\}$ that satisfies
    $\flow'_{s,t}(uv) = \flow'_{t,s}(vu)$ for all
    commodities~$(s,t) \in V^2$
    and arcs~$uv\in A$.
\end{theorem}
\begin{proof}
    We only have to show that~$\flowset$ can be transformed into~$\flowset'$.
    To this end, consider the
    \ac{MCF}~$\bar{\flowset} =
    \{\bar{\flow}_{s,t} \ssep (s,t) \in V^2\}$ \enquote{inverse} to~$\flowset$ where we route each commodity~$(s,t)$ along the reversed arcs of the original flow's commodity~$(t,s)$, i.e., $\bar{\flow}_{s,t}(uv)
    \be \flow_{t,s}(vu)$. $\bar{\flowset}$ is clearly feasible since both directions of each arc
    have the same capacity.
    By averaging both of these \acp{MCF}, we construct
    $\flowset' =  \{\flow'_{s,t} \ssep (s,t) \in V^2\}$ that is not only feasible but also satisfies
    constraint~\eqref{eq:ilp_bidir_flow_symmetry}:
    \begin{align*}
        \flow'_{s,t}(uv) \be %
        \frac{\flow_{s,t}(uv) + \flow_{t,s}(vu)}{2} = \flow'_{t,s}(vu). & \qedhere
    \end{align*}
\end{proof}

\Cref{th:bidir_flow_symmetry} also allows us to drastically reduce the size of the
ILP when implementing it in practice:
Instead of enforcing the constraints~\eqref{eq:ilp_bidir_flow_symmetry}
explicitly, we can assign a unique integer index~$\idx(w)$ to each
node~$w\in V$ and only create the flow variables $\flow_{s,t}(uv)$, $(s,t)\in
V^2$, $uv \in A$ with $\idx(s) \leq \idx(t)$.
Then, we can replace all occurrences of $\flow_{s,t}(uv)$ with $\idx(s) >
\idx(t)$ in the constraints by~$\flow_{t,s}(vu)$ and remove all those
constraints that consequently become superfluous: this not only includes
the constraints~\eqref{eq:ilp_bidir_flow_symmetry}, but also all flow
conservation constraints~\eqref{eq:ilp_bidir_flow_conservation} for
commodities~$(s,t)$ with $st \in A$,  $\idx(s) > \idx(t)$, and capacity
constraints~\eqref{eq:ilp_bidir_cap_constraint} for arcs~$uv \in A$ with
$\idx(u) > \idx(v)$.
However, for the sake of clarity, below we will discuss
ILP~\eqref{eq:ilp_mcfs_undirected} as it is written.

The relaxation of ILP~\eqref{eq:ilp_mcfs_undirected} is obtained by replacing
the integrality constraints~\eqref{eq:ilp_bidir_integrality_constraint} on $x_e$
by the inequalities $0 \leq x_e \leq \paredges{e}$ for all $e \in \edges$.
We will call this LP relaxation \enquote{LP~\eqref{eq:ilp_mcfs_undirected}} for short.
The only other type of constraint that bounds the $x$-variables
is~\eqref{eq:ilp_bidir_cap_constraint}, which yields two lower bounds
$\nicefrac{\flowsum(uv)}{\ccap(\bidir{u,v})}$ and
$\nicefrac{\flowsum(vu)}{\ccap(\bidir{u,v})}$ for each $x_{\bidir{u,v}}$ with
$\bidir{u,v} \in \edges$.
Due to constraint \eqref{eq:ilp_bidir_flow_symmetry}, these lower bounds are equal:
\begin{align*}
    x_{\bidir{u,v}} \geq \frac{\flowsum(uv)}{\ccap(\bidir{u,v})} =
    \frac{\flowsum(vu)}{\ccap(\bidir{u,v})} & & \forall \bidir{u,v} \in \edges.
\end{align*}

Moreover, this bound is always attained exactly as the $x$-variables can
take on fractional values in the LP relaxation and their sum is minimized.
Hence, LP~\eqref{eq:ilp_mcfs_undirected} is equivalent to a classical
\ac{MCF}-LP where the sum of edge utilizations~$\sum_{\bidir{u,v} \in \edges} \frac{\flowsum(uv)}{\ccap(\bidir{u,v})}$ is
minimized.
We thus refer to $\cost(e) \be \frac{1}{\ccap(e)}$ as the \emph{cost} incurred
by sending one unit of flow over edge~$e$.

\subsection{Approximation for Non-uniform Capacities}
Given a bidirected graph $G=(V,A)$ with~$\paredges{e}$ connections per bidirected
edge~$e\in \edges$, let $\minpar \be \min_{e \in \edges}\paredges{e}$. We propose the following algorithm as a
$\max(\frac{1}{\mcfalpha\cdot\minpar},2)$-approximation for \prob{TOCA}:
compute a basic
optimal solution~$\lpopt$ for LP~\eqref{eq:ilp_mcfs_undirected} and round up its
values, i.e., choose $\ceil{\lpopt_e}$ active connections for each edge~$e\in \edges$ ($0$ is still rounded to $0$, of course).
We require $\lpopt$ to be \emph{basic}, i.e.,
it cannot be expressed as a convex combination of two or more other
feasible solutions \cite[p.\ 100]{DBLP:books/daglib/0004338}.
This is not a strong restriction as all standard LP solving algorithms
will always return a basic optimal solution (if any solution
exists)~\cite[p.\ 279]{DBLP:books/daglib/0030297}.

To establish the approximation ratio, we
roughly follow the proof of \cite[Lemma~12 \& Theorem~14]{DBLP:conf/isaac/ChimaniI25}, which proves a
similar guarantee for a closely related rounding scheme w.r.t.\ \prob{MMCFS}.
The main necessity for deviating from their proof is that we minimize the number of
connections per link (instead of binary decisions), and that we only allow decisions on bidirected instead of directed connections.
We are also able to show that the approximation ratio improves when the minimum
number of parallel connections increases.

The key observation allowing for our approximation ratio is that a basic optimal
solution for LP~\eqref{eq:ilp_mcfs_undirected} will only have few
variables $\lpopt_e$ with~$\lpopt_e \in (0,\mcfalpha\cdot\paredges{e})$ in
comparison to the number of \emph{saturated} edges, i.e., edges used
to their full capacity:

\newcommand{\numsat}{\ensuremath{h}}
\newcommand{\numtiny}{\ensuremath{\ell}}
\newcommand{\satedges}{\ensuremath{H}}
\newcommand{\tinyedges}{\ensuremath{L}}
\begin{lemma}\label{th:lp_bidir_relaxation_bound}
    Let~$\lpopt$ be a basic optimal solution for LP~\eqref{eq:ilp_mcfs_undirected},
    and~$\numsat$ the number of edges~$e\in \edges$ such that $\lpopt_e = \paredges{e}$.
    There exist at most $\numsat$ many edges~$e'$ with $\lpopt_{e'} \in (0,
    \mcfalpha\cdot\paredges{e'})$.
\end{lemma}
\begin{proof}
    Let
    $\satedges \be \{e \in \edges \ssep \lpopt_e = \paredges{e}\}$ be the $\numsat$ many edges
    that are saturated by the fractional \ac{MCF}, and $\tinyedges \be \{e \in \edges
    \ssep \lpopt_e \in (0,\mcfalpha \cdot \paredges{e})\}$ with $\numtiny \be
    |\tinyedges|$ the edges that are only used to a fraction less than
    $\mcfalpha\cdot\paredges{e}$;
    in short, edges with \emph{high} and \emph{low} (but non-zero) corresponding~$\lpopt$-values.
    We show that $\numtiny >
    \numsat$ would imply that $\lpopt$ is not basic.

    \newcommand{\tmpvec}{q}
    For every edge~$e=\bidir{s,t} \in \tinyedges$, let $P_{e}$ denote an
    arbitrary \emph{alternative $s$-$t$-path} not using~$e$
    and $\flow_{s,t}(uv) = \flow_{t,s}(vu) > 0$ for all edges $\bidir{u,v} \in
    P_{e}$.
    Such a path must exist to satisfy the demand of
    $\mcfalpha\cdot\alldemand(s,t) = %
    \mcfalpha \cdot \ecap(e)$.
    By optimality of~$\lpopt$, routing a unit of flow over~$P_{e}$ must not be more costly than over $e$ itself,
    i.e., $\sum_{e' \in P_{e}} \frac{1}{\ecap(e')} \leq \frac{1}{\ecap(e)}$.
    We construct a matrix~$M \in
    \{0,1\}^{\numtiny\times\numsat}$ indexed by pairs $(e,e'') \in \tinyedges
    \times \satedges$:
    \begin{align*}
        M(e,e'') =
        \begin{cases*}
            1 & if $e'' \in P_e$,\\
            0 & otherwise.
        \end{cases*}
    \end{align*}
    Since $\numtiny > \numsat$, the $\numtiny$~rows of $M$ must be linearly
    dependent, i.e., there exists a vector $\tmpvec \in
    \Q^{\numtiny}$ with $\tmpvec \neq \mathbf{0}$, such that $\transpose{\tmpvec} \cdot M =
    \mathbf{0}$.

    Based on~$\tmpvec$, we obtain two new feasible solutions by, for each edge~$e \in \tinyedges$, moving a tiny amount of flow from~$e$ to~$P_e$ or vice versa, while adhering to optimality and the edge capacities.
    Formally, for small enough positive $\varepsilon \in \Q$, the
    following vector~$p \in \Q^{|\edges|}$ yields two feasible LP solutions
    $(\lpopt + p)$ and $(\lpopt - p)$:
    \begin{align*}
        p(e') =
        \begin{cases*}
            \varepsilon\cdot\tmpvec(e') -\varepsilon \cdot \sum_{e \in \tinyedges \colon e' \in P_e} \tmpvec(e) & if $e' \in \tinyedges$,\\
            -\varepsilon\cdot \sum_{e \in \tinyedges \colon e' \in P_e}\tmpvec(e) & otherwise. %
        \end{cases*}
    \end{align*}

    Both solutions satisfy all demands and flow conservation
    constraints~\eqref{eq:ilp_bidir_flow_conservation}.
    Consider the capacity constraints~\eqref{eq:ilp_bidir_cap_constraint}:
    By construction, the flow difference on saturated edges is 0
    (for both directions). %
    We only modify flow over a non-saturated edge if it was already non-zero, so a sufficiently small~$\varepsilon$ guarantees a resulting flow between $0$ and the edge's capacity.

    \newcommand{\emax}{\hat{e}}  %
    It remains to show that $p \neq \mathbf{0}$ given that $\tmpvec \neq
    \mathbf{0}$.
    Let $\tinyedges'\subseteq \{e\in \tinyedges : q(e)\neq 0\}$ and $\emax\coloneqq\argmax_{e \in \tinyedges'}\frac{1}{\ecap(e)}$.
    We show that $p(\emax) \neq 0$.
    If~$\emax\in P_e$ for some $e \in \tinyedges'$, then $\sum_{e' \in P_e} \frac{1}{\ecap(e')} \leq
    \frac{1}{\ecap(e)}$ would establish $|P_e| = 1$, and thus require parallel edges.
    Even if we allow parallel edges, shifting small amounts of flow between them yields two alternative solutions contradicting that $\lpopt$ is a basic solution.
    Hence, $\emax$ is only contained in alternative paths~$P_e$ with~$e \in
    \tinyedges\setminus \tinyedges'$, and
\[p(\emax)
        = \varepsilon\cdot\tmpvec(\emax) -\varepsilon \cdot \sum_{e \in
        \tinyedges \colon \emax \in P_e} \tmpvec(e)
        = \varepsilon\cdot\tmpvec(\emax) - 0 \neq 0.\qedhere\]
\end{proof}

Intuitively, by \Cref{th:lp_bidir_relaxation_bound}, there are so few low-value variables
that rounding all of them up increases the objective function only moderately (at most doubling the value contributed by the saturated variables); meanwhile, there can be many more non-saturated variables, but they already hold a significant value in the objective function anyhow, so rounding them up also only contributes a small additional term (depending on $\frac{1}{\mcfalpha}$) to the objective function.
This allows us to obtain the following approximation guarantees, critically deviating from and improving upon~\cite{DBLP:conf/isaac/ChimaniI25}:

\begin{theorem}\label{th:mcfs_bidirected_approx}
    Let~$\lpopt$ be a basic optimal solution for LP~\eqref{eq:ilp_mcfs_undirected}.
    Choosing $\ceil{\lpopt_e}$~many connections into the solution for each
    bidirected edge~$e$ is an $r$-approximation for \prob{TOCA} with
    \begin{align*}
        r \be
        \begin{cases}
            \max(\frac{1}{\mcfalpha\cdot\minpar}, 2) & \text{if } \mcfalpha\cdot\minpar < 1,\\
            1+ \frac{1}{\mcfalpha\cdot\minpar}  & \text{if } \mcfalpha\cdot\minpar \geq 1.
        \end{cases}
    \end{align*}
    where $\minpar \be \min_{e \in \edges}\paredges{e}$.
    Since $r \leq 2$ for $\mcfalpha\cdot\minpar \geq 1$, we can give a
    simplified (but coarse) bound of $r \leq
    \max(\frac{1}{\mcfalpha\cdot\minpar}, 2)$.
\end{theorem}
\begin{proof}
    The solution~$\ceil{\lpopt_e}$ is clearly feasible. Let $z$ be its objective value.
    First, we consider case $\mcfalpha\cdot\minpar < 1$:
    Let
    \begin{align*}
        X &\be \{e \in \edges \ssep \lpopt_e \in (0,\mcfalpha\cdot\minpar)\},\\
        Y_1 &\be \{e \in \edges \ssep \lpopt_e \in \mathopen[\mcfalpha\cdot\minpar,1)\},\\
        Y_2 &\be \{e \in \edges \ssep \lpopt_e \in \mathopen[1,\paredges{e})\}, \text{ and} \\
        Z &\be \{e \in \edges \ssep \lpopt_e = \paredges{e}\}.
    \end{align*}
    By \Cref{th:lp_bidir_relaxation_bound}, we know that
    \begin{align*}
        \sum_{e \in X} \ceil{\lpopt_e} = |X| \leq |Z| = \sum_{e\in Z} \frac{\lpopt_e}{\paredges{e}} \leq \frac{1}{\minpar} \cdot \sum_{e\in Z} \lpopt_e.
    \end{align*}
    For the other edges, we observe:
    \begin{align*}
        \sum_{e \in Y_1} \ceil{\lpopt_e} & %
        \leq \frac{1}{\mcfalpha\cdot\minpar}\cdot \sum_{e\in Y_1} \lpopt_e,\\
        \sum_{e \in Y_2} \ceil{\lpopt_e} &\leq 2 \cdot \sum_{e \in Y_2} \lpopt_e,\\
        \sum_{e \in Z} \ceil{\lpopt_e} &= \sum_{e \in Z} \lpopt_e.
    \end{align*}

    Using the minimum fractional objective value $\lpval$ as a lower bound for
    the minimum integral objective value~$\ilpval$, we can then give an upper
    bound for the approximation ratio:
    \begin{align*}
        \frac{\algval}{\ilpval} &\leq \frac{\algval}{\lpval}
        = \frac{\sum_{e \in \edges} \ceil{\lpopt_e}}{\sum_{e \in \edges} \lpopt_e}
        = \frac{\sum_{e \in X \union Y_1 \union Y_2 \union Z} \ceil{\lpopt_e}}{\sum_{e \in \edges} \lpopt_e}\\
        &\leq \frac{\frac{1}{\mcfalpha\cdot\minpar} \sum_{e \in Y_1} {\lpopt_e} + 2 \sum_{e \in Y_2} {\lpopt_e} + (1+ \frac{1}{\minpar}) \sum_{e \in Z} {\lpopt_e}}{\sum_{e \in \edges} \lpopt_e}\\
        &\leq \frac{\max(\frac{1}{\mcfalpha\cdot\minpar},2) \cdot \sum_{e \in \edges}
        \lpopt_e}{\sum_{e \in \edges} \lpopt_e} =
        \max\left(\frac{1}{\mcfalpha\cdot\minpar},2\right)
    \end{align*}
    as $\minpar\in\N_{\geq 1}$ and thus $(1+ \frac{1}{\minpar}) \leq 2$.

    Now consider the case $\mcfalpha\cdot\minpar \geq 1$ and note that $\edges = X
    \union Y \union Z$ with $Y \be \{e \in \edges \ssep \lpopt_e \in
    \mathopen[\mcfalpha\cdot\minpar,\paredges{e})\}$, and
    \begin{align*}
        \sum_{e \in Y} \ceil{\lpopt_e} &\leq \sum_{e \in Y} (\lpopt_e + 1) \leq (1 + \frac{1}{\mcfalpha\cdot\minpar}) \cdot \sum_{e \in Y} \lpopt_e .
    \end{align*}
    This upper bound, in conjunction with the above bounds w.r.t.\ $X$ and $Z$, yields the approximation
    ratio
    \begin{align*}
        \frac{\algval}{\ilpval} &\leq \frac{\algval}{\lpval}
        = \frac{\sum_{e \in \edges} \ceil{\lpopt_e}}{\sum_{e \in \edges} \lpopt_e}
        = \frac{\sum_{e \in X \union Y \union Z} \ceil{\lpopt_e}}{\sum_{e \in \edges} \lpopt_e}\\
        &\leq \frac{(1+ \frac{1}{\mcfalpha\cdot\minpar}) \sum_{e \in Y} \lpopt_e + (1+ \frac{1}{\minpar}) \sum_{e \in Z} {\lpopt_e}}{\sum_{e \in \edges} \lpopt_e}\\
        &\leq \frac{(1+ \frac{1}{\mcfalpha\cdot\minpar}) \cdot \sum_{e \in \edges}
        \lpopt_e}{\sum_{e \in \edges} \lpopt_e} = 1+ \frac{1}{\mcfalpha\cdot\minpar}.\qedhere
    \end{align*}
\end{proof}

\subsection{Approximation for Uniform Capacities}
In some real-world backbone networks, most links have the same capacity. %
If \emph{all} have the same capacity, \cite[Corollary 7]{DBLP:conf/isaac/ChimaniI25} argues that
simply choosing all arcs into the solution is already a
$\frac{1}{\mcfalpha}$-approximation for \prob{MMCFS}; this proof can be extended
to show that choosing all edges into the solution is an approximation
with the same guarantee for \prob{TOCA}.
However, we show that our rounding strategy improves upon this ratio by
exploiting the number of parallel connections per edge.

\begin{theorem}\label{th:mcfs_bidirected_approx_uniform}
    Let $G=(V,A)$ be a bidirected graph with \emph{uniform} arc capacities~$\ecap$,
    and~$\mcfalpha = \frac{p}{q}$ a retention ratio.
    Choosing exactly~$\ceil{\mcfalpha \cdot \paredges{e}}$ many
    connections for each bidirected edge~$e \in \edges$ into the solution is a
    linear-time
    $\min \{ \frac{1}{\mcfalpha}, (1 + \frac{(q-1)}{p \cdot
    \meanpar})\}$-approximation for \prob{TOCA}, where $\meanpar$ denotes the
    average number of connections per edge.
\end{theorem}
\begin{proof}
An optimal fractional solution~$\{\lpopt_e \ssep e \in \edges\}$ for
LP~\eqref{eq:ilp_mcfs_undirected} will always set $\lpopt_e = \mcfalpha \cdot
\paredges{e}$ for each edge~$e \in \edges$:
all edges have the same capacity, so routing a
demand $\mcfalpha\cdot\alldemand(e)$ over an alternative path (with two or more
edges) would always incur a higher cost than routing it over~$e$ directly.

The solution~$\ceil{\mcfalpha \cdot \paredges{e}}$ is thus feasible and
can be computed in $\bigO(m)$~time.
We obtain the ratio of the algorithm's objective value~$\algval$ and the optimal ILP's
objective value~$\ilpval$ by using the optimal LP's objective
value~$\lpval$ as a lower bound for~$\ilpval$:
\begin{align*}
    \frac{\algval}{\ilpval} &\leq \frac{\algval}{\lpval} = \frac{\sum_{e \in \edges} \ceil{\lpopt_e}}{\sum_{e \in \edges} \lpopt_e}
                    = 1 + \frac{\sum_{e \in \edges} (\ceil{\lpopt_e} - \lpopt_e)}{\sum_{e \in \edges} \lpopt_e}\\
                    &= 1 + \frac{\sum_{e \in \edges} (\ceil{\frac{p}{q} \cdot \paredges{e}} - \frac{p}{q} \cdot \paredges{e})}{\sum_{e \in \edges} (\frac{p}{q} \cdot \paredges{e})}\\
                    &\leq  1 + \frac{\sum_{e \in \edges} (\frac{q-1}{q})}{\sum_{e \in \edges} (\frac{p}{q} \cdot \paredges{e})}
                    = 1 + \frac{|\edges|\cdot(q-1)}{p \cdot |\edges|\cdot\meanpar}.\qedhere
\end{align*}
\end{proof}

\subsection{Dynamic Reconfigurations}

Rather than having only two network topologies, one for the peak times and
one for the low traffic times, one might want to compute several other
topologies for the times in between.
Such configurations for different levels of traffic, or even completely dynamic
reconfigurations, would allow for more fine-grained traffic control (at the
price of additional overhead during actual deployment).

An obvious approach to find reduced topologies for different levels of traffic
would be to compute \prob{TOCA} solutions for the same original network but
different values for~$\mcfalpha$.
However, we note that an optimal solution for a \prob{TOCA}
instance~$(G,\ecap,\mcfalpha_1)$ does not necessarily contain any optimal
solution for the instance~$(G,\ecap,\mcfalpha_2)$ with~$\mcfalpha_2 <
\mcfalpha_1$, even when all arc capacities are 1.
This is true both in the cases of \prob{TOCA} on bidirected graphs with one
connection per edge and \prob{MMCFS} on general directed graphs, as shown by
\Cref{fig:mcfs_monotonicty}.
However, to allow for dynamic reconfigurations in practice,
we may first optimize for~$\mcfalpha_1$, and then
reduce this solution further to obtain a solution
for~$\mcfalpha_2 < \mcfalpha_1$.
Alternatively, we may extend a solution w.r.t.\ $\mcfalpha_2$ to one for~$\mcfalpha_1$.

\begin{figure}[htpb]
    \centering
    \begin{tikzpicture}[scale=1.0]
        \node[label={-90:$s$},dot] (s) at (0,0) {};
        \node[dot] (a) at (1,0) {};
        \node[label={-90:$v$},dot] (b) at (2,0) {};
        \node[dot] (c) at (3,0) {};
        \node[label={-90:$t$},dot] (t) at (4,0) {};

        \begin{scope}[every edge/.style={biedge,ultra thick}]
            \path (s) edge (a);
            \path (a) edge (b);
            \path (b) edge (c);
            \path (c) edge (t);
        \end{scope}

        \path (s) edge[biedge,bend left=85,draw=lightgray!80,very thick] (t);
        \begin{scope}[every edge/.style={biedge, very thick, dashed}]
            \path (s) edge[bend left=45] (b);
            \path (b) edge[bend left=45] (t);
        \end{scope}
    \end{tikzpicture}
    \caption{A network topology showcasing the non-monotonicity of optimal \prob{TOCA}
        solutions w.r.t.\ $\mcfalpha$. All bidirected edges have capacity 1
        and consist of one connection.
        For~$\mcfalpha_1 \in (\frac{1}{2}, \frac{2}{3}\mathclose]$, the unique optimal
        solution consists of all black edges (both solid and dashed).
        For~$\mcfalpha_2 = \frac{1}{2}$, the unique optimal solution consists of all solid edges (both black
        and gray).
        By replacing all edges with arcs directed from left to right we obtain
        an analogous instance that shows the non-monotonicity of \prob{MMCFS} solutions.
    }
    \label{fig:mcfs_monotonicty}
\end{figure}
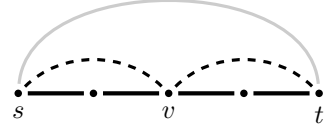

\section{Post-processing}\label{sec:post_processing}

\newcommand\alground{\emph{RND}}
\newcommand\algup{\emph{1-UP}}
\newcommand\algdown{\emph{1-DWN}}
Consider the rounding scheme for \prob{TOCA} presented in
\Cref{sec:approximation}, which amounts to
rounding up a basic optimal solution for
LP~\eqref{eq:ilp_mcfs_undirected}.
We call this scheme \alground{}, and present two
iterative
post-processing heuristics, \algdown{}~(\Cref{alg:mcfs_heur_removal}) and
\algup{}~(\Cref{alg:mcfs_heur_addition}), that further improve \alground{} solutions.
We also tested CPLEX's built-in
feasibility pump, a well-established and thoroughly tested method that
searches for a feasible solution with good objective
value~\cite{DBLP:journals/ejco/Berthold0S19}; it barely
improved \alground{} solutions and thus we do not consider it further.

Before executing the post-processing routines of \algdown{} and \algup{},
we always compute a basic optimal solution~$\lpopt$ for
LP~\eqref{eq:ilp_mcfs_undirected} and then restrict all variables~$x_e$ to
the interval~$\big[\floor{\lpopt_e}, \ceil{\lpopt_e}\big]$.
Not only does this already fix all variables with an integer value; the
constraints~$\{x_{e} \leq \ceil{\lpopt_{e}} \ssep e\in \edges\}$ in particular ensure
that we generate a subset of an \alground{} solution and thus maintain the approximation guarantees given by
\Cref{th:mcfs_bidirected_approx,th:mcfs_bidirected_approx_uniform}.
\algdown{} and \algup{} then execute the following loop until all variables have
integral values:
In each iteration we round one of the variables~$x_e$, $e\in \edges$, with the
smallest rounding difference either down (for \algdown{}) or up (for \algup{}).
By computing LP~\eqref{eq:ilp_mcfs_undirected} again under this restriction, we
obtain a new set of values for our variables.
However, for \algdown{}, this new LP may be infeasible; in this case we dismiss
our previous change to the LP and round up~$x_e$ instead, allowing us to
recompute the LP to obtain new variable values.
Both heuristics will terminate after at most $|\edges|$~iterations,
since in each iteration, at least one additional variable---the one we rounded up or
down---is fixed to an integer value.

\newcommand{\LP}{\ensuremath{\mathcal{L}}}
\begin{algorithm}[tb]
    \caption{Heuristic \algdown{}.}
    \label{alg:mcfs_heur_removal}
    \begin{algorithmic}[1]
        \Input{bidirected graph $G = (V,A)$ with $\ecap$, $\paredges{}$; $\mcfalpha \in (0,1)$.}
        \Let{$\LP$}{LP~\eqref{eq:ilp_mcfs_undirected} for $(G,\ecap,\paredges{},\mcfalpha)$}
        \Let{$\lpopt$}{basic optimal solution of $\LP$}
        \State add constraints $\{\floor{\lpopt_{e}} \leq x_{e} \leq \ceil{\lpopt_{e}} \ssep e\in \edges\}$ to $\LP$
        \While{$\lpopt$ is not integral}
            \newcommand{\tmpvar}{\tilde{x}_e}
            \Let{$e$}{$\argmin_{e' \in \edges \colon \lpopt_{e'} \notin \N} (\lpopt_{e'} - \floor{\lpopt_{e'}})$}
            \Let{$\tmpvar$}{$\lpopt_e$}
            \State add constraint $x_e = \floor{\tmpvar}$ to \LP
            \Let{$\lpopt$}{optimal solution of \LP}
            \If{$\lpopt$ is infeasible}
                \State replace constraint $x_e = \floor{\tmpvar}$ by $x_e = \ceil{\tmpvar}$ in $\LP$
                \Let{$\lpopt$}{optimal solution of $\LP$}
            \EndIf
        \EndWhile
        \Return $\lpopt$
    \end{algorithmic}
\end{algorithm}

\begin{algorithm}[tb]
    \caption{Heuristic \algup{}.}
    \label{alg:mcfs_heur_addition}
    \begin{algorithmic}[1]
        \Input{bidirected graph $G = (V,A)$ with $\ecap$, $\paredges{}$; $\mcfalpha \in (0,1)$.}
        \Let{$\LP$}{LP~\eqref{eq:ilp_mcfs_undirected} for $(G,\ecap,\paredges{},\mcfalpha)$}
        \Let{$\lpopt$}{basic optimal solution of $\LP$}
        \State add constraints $\{\floor{\lpopt_{e}} \leq x_{e} \leq \ceil{\lpopt_{e}} \ssep e\in \edges\}$ to $\LP$
        \While{$\lpopt$ is not integral}
        \Let{$e$}{$\argmin_{e' \in \edges \colon \lpopt_{e'} \notin \N} (\ceil{\lpopt_{e'}} - \lpopt_{e'})$}
            \State add constraint $x_e = \ceil{\lpopt_e}$ to \LP
            \Let{$\lpopt$}{optimal solution of \LP}
        \EndWhile
        \Return $\lpopt$
    \end{algorithmic}
\end{algorithm}

\section{Experimental Setup}\label{sec:evaluation_setup}

\begin{figure*}
\centering
\begin{subfigure}{.32\textwidth}
  \centering
  \includegraphics[width=\linewidth]{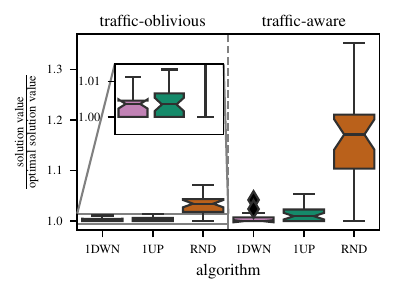}
  \caption{Solution value comparison between the rounding algorithms and the ILP.}
  \label{fig:solutionValue}
\end{subfigure}%
\hfill
\begin{subfigure}{.32\textwidth}
  \centering
  \includegraphics[width=\linewidth]{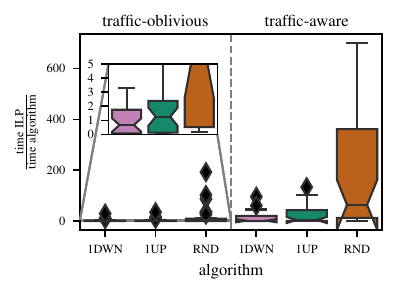}
  \caption{Computation time comparison between the rounding algorithms and the ILP.}
  \label{fig:Time}
\end{subfigure}%
\hfill
\begin{subfigure}{.32\textwidth}
  \centering
  \includegraphics[width=\linewidth]{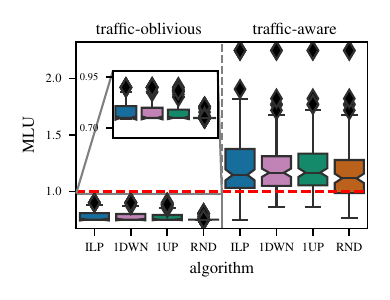}
  \caption{\acf{MLU} of all approaches.}
  \label{fig:MLU}
\end{subfigure}
\caption{Results of the experimental evaluation}
\label{fig:Results}
\end{figure*}

We provide an overview of our experimental setup, including the evaluated algorithms, their parametrization, and the dataset.
All experiments were conducted on a system equipped with two AMD EPYC 7452 CPUs and 512~GB of RAM, running 64-bit Ubuntu 20.04.1.
All LPs and ILPs were solved using CPLEX~\cite{cplex}.

\subsection{Algorithms and Parametrization}
We evaluate four algorithms:
the LP-rounding scheme \alground{} given by \Cref{th:mcfs_bidirected_approx}
and its two post-processing heuristics \algup{} and \algdown{} (presented in \Cref{sec:post_processing}), as well as an exact algorithm that simply solves ILP~\eqref{eq:ilp_mcfs_undirected} optimally and thus serves as a baseline.
Each of these algorithms has two variants:
The \textit{traffic-oblivious} variant only considers the
scaled-down traffic matrix~$\mcfalpha\cdot\alldemand$ where~$\alldemand$ is given by \Cref{eq:super_t_arcs} in \Cref{subsec:Problem}.
In contrast, the \emph{traffic-aware} variant optimizes the network w.r.t.\ to a scaled-down traffic matrix~$\mcfalpha\cdot\demand$ where~$\demand$ is one specific %
traffic matrix provided by the dataset; here, the approximation guarantees of \Cref{th:mcfs_bidirected_approx,th:mcfs_bidirected_approx_uniform} do not hold.

Evaluating both the traffic-oblivious and traffic-aware variants
allows us to compare how many active connections the respective output topologies contain, and how well they are able to accommodate differently structured traffic during low-traffic periods (such as the night).
To give a precise measure of the latter performance criterion, we
take traffic matrices from our dataset, scale them down by $\mcfalpha$, route them through the output topologies using the 2-\ac{SR}
approach proposed by Bhatia et al.~\cite{DBLP:conf/infocom/BhatiaHKL15}, and compute the resulting \ac{MLU}.

Throughout the whole evaluation, both for the algorithms and the \ac{MLU} computation, we use the retention ratio~$\mcfalpha= 0.5$.
This scaling factor is based on traffic measurements from a major European Tier-1 \ac{ISP}:
as shown in \Cref{FIG:2022-traffic}, network utilization typically drops below 50\% between 1:00 AM and 9:00 AM compared to peak hours in the evening.
We also performed experiments for~$\mcfalpha= 0.3$ and~$\mcfalpha= 0.7$---however, as the corresponding results showed similar patterns to those for~$\mcfalpha= 0.5$, we omit them here due to space limitations.

\subsection{Dataset}
We use the Repetita dataset~\cite{repetita}, a well-known and commonly used \ac{TE} benchmarking dataset.
However, many of the Repetita topologies---such as the original ARPANET---are relatively small, and do not reflect the structure of modern backbone networks.
Thus, we only run our experiments on topologies of the Topology Zoo~\cite{topo-zoo} with more than 60 nodes.

In all topologies of the dataset, each pair of connected nodes is only connected via a single (bidirected) edge.
However, in real backbone networks, it is common to have multiple parallel connections between two routers.
We thus assume that there are $\paredges{e} = 5$ connections for all edges~$e$ in all topologies of the dataset, a value that is
based on the topology of a major European Tier-1 \ac{ISP}.
This allows a granular deactivation of individual connections while avoiding hard disconnects.

The dataset also includes four different traffic matrices per topology, each designed to reflect high-utilization scenarios.
These traffic matrices are used both for the traffic-aware algorithm variants and for the \ac{MLU} computations.

\section{Results}\label{sec:evaluation_results}

We discuss the results of our evaluation, comparing the different algorithms w.r.t.\ three measures: the number of active connections in the output topologies, the computation times, and
the \ac{MLU} yielded by 2-\ac{SR} in the output topologies.

\subsection{Number of Active Connections (Solution Value)}
First we examine the solution values, i.e., the number of active connections in the output topologies, of our rounding algorithms \alground{}, \algdown{}, and \algup{}.
Figure~\ref{fig:solutionValue} compares these to the optimal solution obtained from ILP~\eqref{eq:ilp_mcfs_undirected}.
By plugging the values~$\mcfalpha = 0.5$ and~$\minpar = 5$ (see \Cref{sec:evaluation_setup}) into the approximation guarantee of \Cref{th:mcfs_bidirected_approx}, we know that all three rounding algorithms
are guaranteed to deviate from the optimum by a factor of at most $1.4$ in the traffic-oblivious case.
However, the algorithms not only adhere to this theoretical guarantee, but outperform it by always staying below a factor of $1.1$.
In fact, both \algdown{} and \algup{} yield results that are nearly optimal.
In contrast, \alground{} without post-processing performs significantly worse: the number of active connections is up to 10\% higher in the traffic-oblivious case and up to 30\% higher in the traffic-aware case.

These results highlight two key insights:
First, our sophisticated post-processing heuristics clearly pay off, especially in traffic-aware scenarios.
Second, on instances that resemble real-world networks, they not only adhere to the approximation guarantee but return near-optimal results.
Networks where the guarantee is met exactly typically do not occur in practice since they are already difficult to construct theoretically~\cite{DBLP:conf/isaac/ChimaniI25}.

Further, the traffic-aware approaches are able to deactivate up to 30\% more connections compared to the respective traffic-oblivious variants.
However, they have the disadvantage of only guaranteeing feasibility for one specific traffic matrix.

\subsection{Computation Time}
Figure~\ref{fig:Time} compares the computation time of the rounding algorithms with those of the optimal ILP approach, the latter requiring on average 14 (43) minutes for the traffic-oblivious (traffic-aware, resp.) case. By far the fastest algorithm is \alground{} without post-processing. In the traffic-oblivious case, it is on average roughly 10 times faster than the ILP. This effect is even more evident in the traffic-aware case, where \alground{} is roughly 150 times faster.
It is not surprising that the simple rounding strategy outperforms all other algorithms in terms of computation time---however, it is remarkable that it does so while still having the approximation guarantee of \Cref{th:mcfs_bidirected_approx}.

Evaluating \algdown{} and \algup{} requires a more nuanced view: On one hand, if the computation of the ILP takes less than ten minutes, usually the ILP is faster than both post-processing heuristics.
After all, each iteration of the post-processing involves solving an LP (or even two LPs for some iterations of \algdown{}), which can be quite time consuming.
On the other hand, if the ILP computation exceeds ten minutes, it usually also exceeds an hour, and both post-processing heuristics are significantly faster in comparison.
We thus recommend using \algdown{} or \algup{} if the optimal solution cannot be computed within ten minutes.
Since both of these algorithms give very similar results, we do not have a clear preference for either of them, but \algup{} is sometimes slightly faster.

\subsection{Maximum Link Utilization (MLU)}

Recall that the dataset contains a set of four traffic matrices for each topology.
For the traffic-oblivious algorithm variants, \Cref{fig:MLU} shows the \ac{MLU} of all four traffic matrices.
In contrast, the traffic-aware variants optimize their solution w.r.t.\ one specific traffic matrix out of this set.
For these variants, the plot shows the \ac{MLU} for the three remaining traffic matrices.

All traffic-oblivious algorithm variants successfully handle the evaluated traffic matrices, never exceeding an \ac{MLU} of~1. Among these approaches, the ILP yields the highest \ac{MLU}, with all results clustering around an \ac{MLU} of roughly 0.9.
This is not surprising as the traffic matrices in the Repetita dataset are designed to be challenging in terms of link utilization:
even the best possible routing cannot achieve an \ac{MLU} significantly below 0.9. Therefore, we would expect a similar \ac{MLU} in our low-traffic scenario where both the topology and the traffic matrix have been reduced proportionally.

Further, the \ac{MLU} increases with a lower number of active connections:
The simple rounding algorithm \alground{} without post-processing performs better in terms of \ac{MLU} than the post-processing heuristics \algdown{} and \algdown{}.
In turn, these perform slightly better than the ILP.
This is because a higher number of active connections allows the load to be distributed more evenly, leading to marginally lower \ac{MLU} values.

The results of our \ac{MLU} experiment for the traffic-aware algorithm variants reveal additional critical insights.
None of the traffic-aware solutions is able to reliably handle traffic matrices other than the specific one they were optimized for---even though all matrices have the same total volume and only differ w.r.t.\ how the traffic is distributed among source-destination pairs.
This sensitivity illustrates a key limitation: traffic-aware solutions offer better performance under known conditions but lack robustness when faced with variations in traffic.
Such variations can occur frequently in real-world networks, for example due to routing changes triggered by external failures or altered peering agreements.
Hence, it is of great practical importance to evaluate solutions not only under optimal conditions but also across a diverse set of traffic scenarios.
Moreover, the shortcomings of the traffic-aware algorithm variants emphasize the practical value of the traffic-oblivious approaches.

\section{Conclusion and Future Work}
We presented a novel approach to Green \acl{TE}: instead of using explicit traffic matrices and frequent reconfigurations, we proposed a traffic-oblivious approach that guarantees routability for any down-scaled traffic matrix. We defined the NP-hard problem of activating as few connections as possible while still allowing a \acl{MCF} routing of any such matrix. Further, an approximation algorithm and two post-processing heuristics to improve solution quality were developed.
Our evaluation shows that these algorithms compute near-optimal solutions within low computation times.

In this work, we have laid strong theoretical and practical groundwork for a traffic-oblivious Green \acl{TE} approach. Nevertheless, several challenges remain: First, our approach currently abstracts away from specific hardware. Adapting it to real transponder and router hardware, including their exact power models and routing constraints, will require further theoretical considerations---for instance, introducing a weighted model to reflect the actual energy savings potential of each connection.
Moreover, adapting our methods to \acl{SR} requires additional research. Due to the limited number of segments, our problem reformulation using a single traffic matrix~$\alldemand$ is no longer feasible. Furthermore, the unsplittable flows of \acl{SR} would introduce new integer path variables, making it unclear how to efficiently select single paths.

\newpage
\bibliographystyle{IEEEtranS}
\bibliography{main}

% Generated by IEEEtranS.bst, version: 1.12 (2007/01/11)
\begin{thebibliography}{10}
\providecommand{\url}[1]{#1}
\csname url@samestyle\endcsname
\providecommand{\newblock}{\relax}
\providecommand{\bibinfo}[2]{#2}
\providecommand{\BIBentrySTDinterwordspacing}{\spaceskip=0pt\relax}
\providecommand{\BIBentryALTinterwordstretchfactor}{4}
\providecommand{\BIBentryALTinterwordspacing}{\spaceskip=\fontdimen2\font plus
\BIBentryALTinterwordstretchfactor\fontdimen3\font minus \fontdimen4\font\relax}
\providecommand{\BIBforeignlanguage}[2]{{%
\expandafter\ifx\csname l@#1\endcsname\relax
\typeout{** WARNING: IEEEtranS.bst: No hyphenation pattern has been}%
\typeout{** loaded for the language `#1'. Using the pattern for}%
\typeout{** the default language instead.}%
\else
\language=\csname l@#1\endcsname
\fi
#2}}
\providecommand{\BIBdecl}{\relax}
\BIBdecl

\bibitem{DBLP:conf/ifip6-3/AddisCCGS12}
B.~Addis, A.~Capone, G.~Carello, L.~G. Gianoli, and B.~Sans{\`{o}}, ``Energy aware management of resilient networks with shared protection,'' in \emph{Proc.\ {SustainIT}}.\hskip 1em plus 0.5em minus 0.4em\relax {IEEE}, 2012.

\bibitem{DBLP:conf/stacs/Al-NajjarBL22}
Y.~Al{-}Najjar, W.~Ben{-}Ameur, and J.~Leguay, ``Approximability of robust network design: The directed case,'' in \emph{Proc.\ {STACS}}, ser. LIPIcs, vol. 219.\hskip 1em plus 0.5em minus 0.4em\relax Schloss Dagstuhl, 2022.

\bibitem{DBLP:conf/wowmom/AmaldiCGM11}
E.~Amaldi, A.~Capone, L.~G. Gianoli, and L.~Mascetti, ``Energy management in {IP} traffic engineering with shortest path routing,'' in \emph{Proc.\ {WOWMOM}}.\hskip 1em plus 0.5em minus 0.4em\relax {IEEE} Computer Society, 2011.

\bibitem{DBLP:conf/soda/AndoniGK14}
A.~Andoni, A.~Gupta, and R.~Krauthgamer, ``Towards ({1} + $\epsilon$)-approximate flow sparsifiers,'' in \emph{Proc.\ {SODA}}.\hskip 1em plus 0.5em minus 0.4em\relax {SIAM}, 2014, pp. 279--293.

\bibitem{DBLP:journals/jcss/AzarCFKR04}
Y.~Azar, E.~Cohen, A.~Fiat, H.~Kaplan, and H.~R{\"{a}}cke, ``Optimal oblivious routing in polynomial time,'' \emph{J. Comput. Syst. Sci.}, vol.~69, no.~3, pp. 383--394, 2004.

\bibitem{ben2005routing}
W.~Ben{-}Ameur and H.~Kerivin, ``Routing of uncertain traffic demands,'' \emph{Optimization and Engineering}, vol.~6, pp. 283--313, 2005.

\bibitem{DBLP:journals/ejco/Berthold0S19}
T.~Berthold, A.~Lodi, and D.~Salvagnin, ``Ten years of feasibility pump, and counting,'' \emph{{EURO} J. Comput. Optim.}, vol.~7, no.~1, 2019.

\bibitem{DBLP:conf/infocom/BhatiaHKL15}
R.~Bhatia, F.~Hao, M.~S. Kodialam, and T.~V. Lakshman, ``Optimized network traffic engineering using segment routing,'' in \emph{Proc.\ {INFOCOM}}.\hskip 1em plus 0.5em minus 0.4em\relax {IEEE}, 2015, pp. 657--665.

\bibitem{DBLP:conf/ipco/Bley05}
A.~Bley, ``On the approximability of the minimum congestion unsplittable shortest path routing problem,'' in \emph{Proc.\ {IPCO}}, ser. LNCS, vol. 3509.\hskip 1em plus 0.5em minus 0.4em\relax Springer, 2005, pp. 97--110.

\bibitem{MO_INFOCOM}
A.~Brundiers, T.~Schüller, and N.~Aschenbruck, ``Midpoint optimization for segment routing,'' in \emph{Proc.\ {INFOCOM}}, 2022.

\bibitem{INFOCOM:Efficient_Algorithm}
E.~Bérczi-Kovács, P.~Gyimesi, B.~Vass, and J.~Tapolcai, ``Efficient algorithm for region-disjoint survivable routing in backbone networks,'' in \emph{Proc.\ {INFOCOM}}, 2024, pp. 951--960.

\bibitem{green_sr}
R.~Carpa, O.~Gl\"uck, and L.~Lefevre, ``{Segment Routing based Traffic Engineering for Energy Efficient Backbone Networks},'' in \emph{Proc.\ {ANTS}}, 2014.

\bibitem{DBLP:conf/icc/ChiaraviglioMN09}
L.~Chiaraviglio, M.~Mellia, and F.~Neri, ``Reducing power consumption in backbone networks,'' in \emph{Proc.\ {ICC}}.\hskip 1em plus 0.5em minus 0.4em\relax {IEEE}, 2009.

\bibitem{DBLP:conf/isaac/ChimaniI25}
M.~Chimani and M.~Ilsen, ``Traffic-oblivious multi-commodity flow network design,'' in \emph{Proc.\ {ISAAC}}, ser. LIPIcs, vol. 359.\hskip 1em plus 0.5em minus 0.4em\relax Schloss Dagstuhl, 2025, pp. 19:1--19:17.

\bibitem{INFOCOM:Designing_Optimal_Compact}
K.~Chitavisutthivong, C.~So-In, and S.~Supittayapornpong, ``Designing optimal compact oblivious routing for datacenter networks in polynomial time,'' in \emph{Proc.\ {INFOCOM}}, 2023.

\bibitem{An_Energy_Saving_Routing_Algorithm}
A.~Cianfrani, V.~Eramo, M.~Listanti, M.~Marazza, and E.~Vittorini, ``An energy saving routing algorithm for a green ospf protocol,'' in \emph{Proc.\ {INFOCOM WKSHPS}}, 2010.

\bibitem{OSPF_enhancement}
A.~Cianfrani \emph{et~al.}, ``{An OSPF enhancement for energy saving in IP networks},'' in \emph{Proc.\ {INFOCOM WKSHPS}}, 2011.

\bibitem{INFOCOM:Approximation_Algs}
W.~Dai, M.~Dinitz, K.-T. Foerster, L.~Luo, and S.~Schmid, ``Approximation algorithms for minimizing congestion in demand-aware networks,'' in \emph{Proc.\ {INFOCOM}}, 2024, pp. 1461--1470.

\bibitem{GOSPF}
M.~D'Arienzo and S.~Romano, ``Gospf: An energy efficient implementation of the ospf routing protocol,'' \emph{Journal of Network and Computer Applications}, 2016.

\bibitem{CATE}
S.~El-Zahr, P.~Gunning, and N.~Zilberman, ``Exploring the benefits of carbon-aware routing,'' in \emph{Proc.\ {CoNEXT}}, 2023.

\bibitem{DBLP:journals/coap/FortzT04}
B.~Fortz and M.~Thorup, ``Increasing internet capacity using local search,'' \emph{Comput. Optim. Appl.}, vol.~29, no.~1, pp. 13--48, 2004.

\bibitem{repetita}
\BIBentryALTinterwordspacing
S.~Gay, P.~Schaus, and S.~Vissicchio, ``{{REPETITA:} Repeatable Experiments for Performance Evaluation of Traffic-Engineering Algorithms},'' \emph{CoRR}, vol. abs/1710.08665, 2017. [Online]. Available: \url{http://arxiv.org/abs/1710.08665}
\BIBentrySTDinterwordspacing

\bibitem{DBLP:conf/ict/GhumanN17}
K.~S. Ghuman and A.~Nayak, ``Per-packet based energy aware segment routing approach for data center networks with {SDN},'' in \emph{Proc.\ {ICT}}.\hskip 1em plus 0.5em minus 0.4em\relax {IEEE}, 2017.

\bibitem{DBLP:journals/talg/HajiaghayiKRL07}
M.~T. Hajiaghayi, R.~D. Kleinberg, H.~R{\"{a}}cke, and T.~Leighton, ``Oblivious routing on node-capacitated and directed graphs,'' \emph{{ACM} Trans. Algorithms}, vol.~3, no.~4, p.~51, 2007.

\bibitem{DBLP:conf/cp/HartertSVB15}
R.~Hartert, P.~Schaus, S.~Vissicchio, and O.~Bonaventure, ``Solving segment routing problems with hybrid constraint programming techniques,'' in \emph{Proc.\ {CP}}, ser. LNCS, vol. 9255.\hskip 1em plus 0.5em minus 0.4em\relax Springer, 2015, pp. 592--608.

\bibitem{DBLP:conf/infocom/HassidimRSS13}
A.~Hassidim, D.~Raz, M.~Segalov, and A.~Shaqed, ``Network utilization: The flow view,'' in \emph{Proc.\ {INFOCOM}}.\hskip 1em plus 0.5em minus 0.4em\relax {IEEE}, 2013, pp. 1429--1437.

\bibitem{cplex}
{IBM}, ``{IBM ILOG CPLEX Optimization Studio 20.1.0},'' \url{https://www.ibm.com/docs/en/icos/20.1.0}, 2020.

\bibitem{The_internet}
R.~Jacob and L.~Vanbever, ``The internet of tomorrow must sleep more and grow old,'' \emph{SIGENERGY Energy Inform. Rev.}, 2023.

\bibitem{DBLP:journals/cn/JiaJGSW18}
X.~Jia, Y.~Jiang, Z.~Guo, G.~Shen, and L.~Wang, ``Intelligent path control for energy-saving in hybrid {SDN} networks,'' \emph{Comput. Networks}, vol. 131, pp. 65--76, 2018.

\bibitem{topo-zoo}
S.~Knight, H.~Nguyen, N.~Falkner, R.~Bowden, and M.~Roughan, ``{The Internet Topology Zoo},'' \emph{IEEE Journal on Selected Areas in Communications}, vol.~29, no.~9, pp. 1765--1775, 2011.

\bibitem{GreyWolf}
K.~Kurroliya, S.~Mohanty, K.~Kanodia, and B.~Sahoo, ``{Grey Wolf Aware Energy-saving and Load-balancing in Software Defined Networks Considering Real Time Traffic},'' in \emph{Proc.\ {ICICT}}, 2020.

\bibitem{DBLP:conf/stoc/LeightonM10}
F.~T. Leighton and A.~Moitra, ``Extensions and limits to vertex sparsification,'' in \emph{Proc.\ {STOC}}.\hskip 1em plus 0.5em minus 0.4em\relax {ACM}, 2010, pp. 47--56.

\bibitem{DBLP:conf/focs/Moitra09}
A.~Moitra, ``Approximation algorithms for multicommodity-type problems with guarantees independent of the graph size,'' in \emph{Proc.\ {FOCS}}.\hskip 1em plus 0.5em minus 0.4em\relax {IEEE} Computer Society, 2009, pp. 3--12.

\bibitem{DBLP:conf/icccnt/OkonorGA16}
O.~Okonor, S.~Georgoulas, and M.~Abdullahi, ``Disruption/time-aware green traffic engineering for core {ISP} networks,'' in \emph{Proc.\ {ICCCNT}}.\hskip 1em plus 0.5em minus 0.4em\relax {ACM}, 2016.

\bibitem{DBLP:journals/ijcomsys/OkonorWGS17}
O.~Okonor, N.~Wang, S.~Georgoulas, and Z.~Sun, ``Dynamic link sleeping reconfigurations for green traffic engineering,'' \emph{Int. J. Commun. Syst.}, vol.~30, no.~9, 2017.

\bibitem{1990OrponenApprox}
P.~Orponen and H.~Mannila, ``On approximation preserving reductions: Complete problems and robust measures (revised version),'' Department of Computer Science, University of Helsinki, Tech. Rep. C-1987-28, May 1990.

\bibitem{OttenLCN24}
D.~Otten and N.~Aschenbruck, ``Failure resilient green traffic engineering,'' in \emph{Proc.\ {LCN}}, 2024.

\bibitem{DBLP:conf/lcn/OttenICA23}
D.~Otten, M.~Ilsen, M.~Chimani, and N.~Aschenbruck, ``Green traffic engineering by line card minimization,'' in \emph{Proc.\ {LCN}}.\hskip 1em plus 0.5em minus 0.4em\relax {IEEE}, 2023.

\bibitem{Hypnos}
L.~R\"{o}llin, R.~Jacob, and L.~Vanbever, ``A sleep study for isp networks: Evaluating link sleeping on real world data,'' \emph{SIGENERGY Energy Inform. Rev.}, 2025.

\bibitem{DBLP:journals/jocnet/SankaranS14}
G.~C. Sankaran and K.~M. Sivalingam, ``Load-dependent power-efficient passive optical network architectures,'' \emph{{JOCN}}, vol.~6, no.~12, pp. 1104--1114, 2014.

\bibitem{DBLP:conf/lcn/SchullerACHS17}
T.~Sch{\"{u}}ller, N.~Aschenbruck, M.~Chimani, M.~Horneffer, and S.~Schnitter, ``Predictive traffic engineering with 2-segment routing considering requirements of a carrier {IP} network,'' in \emph{Proc.\ {LCN}}, 2017, pp. 667--675.

\bibitem{INFOCOM:Optimal_Oblivious}
S.~Supittayapornpong, P.~Namyar, M.~Zhang, M.~Yu, and R.~Govindan, ``Optimal oblivious routing for structured networks,'' in \emph{Proc.\ {INFOCOM}}, 2022, pp. 1988--1997.

\bibitem{telefonica}
\BIBentryALTinterwordspacing
Telefonica, ``{Management and Sustainability Report 2023},'' accessed 04.07.2025. [Online]. Available: \url{https://www.telefonica.com/en/wp-content/uploads/sites/5/2024/03/management-sustainability-report-2023.pdf}
\BIBentrySTDinterwordspacing

\bibitem{Identifying_energy_critical_paths}
N.~Vasi\'{c}, P.~Bhurat, D.~Novakovi\'{c}, M.~Canini, S.~Shekhar, and D.~Kosti\'{c}, ``Identifying and using energy-critical paths,'' in \emph{Proc.\ {CoNext}}, 2011.

\bibitem{DBLP:books/daglib/0004338}
\BIBentryALTinterwordspacing
V.~V. Vazirani, \emph{Approximation algorithms}.\hskip 1em plus 0.5em minus 0.4em\relax Springer, 2001. [Online]. Available: \url{http://www.springer.com/computer/theoretical+computer+science/book/978-3-540-65367-7}
\BIBentrySTDinterwordspacing

\bibitem{INFOCOM:MulticasTE}
C.-H. Wang, S.-H. Chiang, S.-H. Shen, D.-N. Yang, and W.-T. Chen, ``Multicast traffic engineering with segment trees in software-defined networks,'' in \emph{Proc.\ {INFOCOM}}, 2020, pp. 1808--1817.

\bibitem{DBLP:books/daglib/0030297}
\BIBentryALTinterwordspacing
D.~P. Williamson and D.~B. Shmoys, \emph{The Design of Approximation Algorithms}.\hskip 1em plus 0.5em minus 0.4em\relax Cambridge University Press, 2011. [Online]. Available: \url{http://www.cambridge.org/de/knowledge/isbn/item5759340/}
\BIBentrySTDinterwordspacing

\bibitem{DBLP:conf/icnp/ZhangYLZ10}
M.~Zhang, C.~Yi, B.~Liu, and B.~Zhang, ``{GreenTE}: Power-aware traffic engineering,'' in \emph{Proc.\ {ICNP}}.\hskip 1em plus 0.5em minus 0.4em\relax {IEEE} Computer Society, 2010, pp. 21--30.

\end{thebibliography}

\end{document}